\documentclass[a4, 11pt]{article}

\usepackage{fullpage}
\usepackage[lined, algoruled, noend,linesnumbered]{algorithm2e}
\usepackage{xcolor,soul}
\usepackage{amsmath,amssymb,amsthm}
\usepackage{thm-restate, thmtools}
\usepackage{hyperref}
\usepackage{graphicx}

\title{Efficient Contractions of Dynamic Graphs – with Applications}
\author{
    Monika Henzinger\thanks{Institute of Science and Technology, Klosterneuburg, Austria}
    \and
    Evangelos Kosinas\footnotemark[1]
    \and
    Robin Münk\thanks{Technical University of Munich, Germany}
    \and
    Harald Räcke\footnotemark[2]
}

\definecolor{cblue}{rgb}{0.36, 0.54, 0.66}
\definecolor{cred}{rgb}{1, 0, 0}

\def\tildeO{\tilde{O}}
\def\hatO{\hat{O}}
\def\hatG{\hat{G}}

\def\compactpara #1{\textbf{#1}}
\newcommand{\ea}{et al.}

\def\USF{U_{\mathrm{SF}}}
\def\UCS{U_{\mathrm{CS}}}

\def\Epre{E_{\text{pre}}}
\def\Echeck{E_{\text{check}}}
\def\Econ{E_{\text{con}}}

\def\poly{\operatorname{poly}}
\def\polylog{\operatorname{polylog}}

\newtheorem{theorem}{Theorem}
\newtheorem{lemma}[theorem]{Lemma}
\newtheorem{corollary}[theorem]{Corollary}

\begin{document}

\bibliographystyle{plainurl}

\maketitle

\begin{abstract}
A non-trivial minimum cut (NMC) sparsifier is a multigraph $\hat{G}$ that preserves all non-trivial minimum cuts of a given undirected graph $G$.
We introduce a flexible data structure for fully dynamic graphs that can efficiently provide an NMC sparsifier upon request at any point during the sequence of updates. We employ simple dynamic forest data structures to achieve a fast from-scratch construction of the sparsifier at query time.
Based on the strength of the adversary and desired type of time bounds, the data structure comes with different guarantees.
Specifically, let $G$ be a fully dynamic simple graph with $n$ vertices and minimum degree $\delta$. Then our data structure supports an insertion/deletion of an edge to/from $G$ in $n^{o(1)}$ worst-case time. Furthermore, upon request, it can return w.h.p. an NMC sparsifier of $G$ that has $O(n/\delta)$ vertices and $O(n)$ edges, in $\hat{O}(n)$ time. The probabilistic guarantees hold against an adaptive adversary.
Alternatively, the update and query times can be improved to $\tilde{O}(1)$ and $\tilde{O}(n)$ respectively, if amortized-time guarantees are sufficient, or if the adversary is oblivious \footnotemark.
\footnotetext{Throughout the paper, we use $\tildeO$ to hide polylogarithmic factors and $\hatO$ to hide subpolynomial (i.e., $n^{o(1)}$) factors.}

We discuss two applications of our new data structure.
First, it can be used to efficiently report a cactus representation of all minimum cuts of a fully dynamic simple graph. Building this cactus for the NMC sparsifier instead of the original graph allows for a construction time that is sublinear in the number of edges. Against an adaptive adversary, we can with high probability output the cactus representation in worst-case $\hatO(n)$ time.
Second, our data structure allows us to efficiently compute the maximal $k$-edge-connected subgraphs of undirected simple graphs, by repeatedly applying a minimum cut algorithm on the NMC sparsifier. Specifically, we can compute with high probability the maximal $k$-edge-connected subgraphs of a simple graph with $n$ vertices and $m$ edges in $\tilde{O}(m+n^2/k)$ time. This improves the best known time bounds for $k = \Omega(n^{1/8})$ and naturally extends to the case of fully dynamic graphs.
\end{abstract}

\section{Introduction}\label{sec:introduction}
Graph sparsification is an algorithmic technique that replaces an input graph $G$ by another graph  $\hatG$, which has fewer edges and/or vertices than $G$, but preserves (or approximately preserves) a desired graph property. Specifically, for connectivity-based and flow-based problems, a variety of static sparsifiers exist, that approximately maintain cut- or flow-values in $G$~\cite{BSS12,DBLP:conf/stoc/BenczurK96, DBLP:journals/siamcomp/BenczurK15, DBLP:journals/siamcomp/FungHHP19,LS18,SS11,ST11}.

Let $n$ denote the number of vertices, $m$ the number of edges, and $\delta$ the minimum degree of the input graph.
In an undirected {simple} graph it is possible to reduce the number of vertices to $O(n/\delta)$ and the number of edges to $O(n)$ both in randomized and deterministic time $\tilde O(m)$ while preserving the value of all non-trivial minimum cuts \emph{exactly}~\cite{DBLP:conf/soda/Ghaffari0T20,DBLP:journals/jacm/KawarabayashiT19,DBLP:journals/dam/LoST21}.
A minimum cut is considered \emph{trivial} if one of its sides consists of a single vertex and we call the resulting multigraph a \emph{non-trivial minimum cut sparsifier (NMC sparsifier)}.

Most sparsification algorithms assume a static graph. Maintaining an NMC sparsifier in a fully dynamic setting, where a sequence of edge insertions and deletions can be arbitrarily interleaved with requests to output the NMC sparsifier, was only recently studied: Goranci, Henzinger, Nanongkai, Saranurak, Thorup, and Wulff-Nilsen~\cite{DBLP:conf/soda/GoranciHNSTW23} show how to maintain an NMC sparsifier in a fully dynamic graph w.h.p. in $\tilde O(n)$ worst-case update and query time (as a consequence of Theorem 3.7 in~\cite{DBLP:conf/soda/GoranciHNSTW23}), under the assumption of an oblivious adversary.
Additionally, Theorem 4.5 in~\cite{DBLP:conf/soda/GoranciHNSTW23} gives a deterministic algorithm that outputs an NMC sparsifier of size $\tilde O(m/\delta)$ in $\tilde O(m/\delta)$ worst-case query time with $\tilde O(\delta^3)$ amortized update time.

\subsection{Our Results}
In this paper, we present the first data structure for providing an NMC sparsifier of a fully dynamic graph that supports sublinear \emph{worst-case} update and query time and works against an \emph{adaptive} adversary. As a first application, we give an improved fully dynamic algorithm that outputs a cactus representation of \emph{all} minimum cuts of the current graph upon request.
Additionally, we use our data structure to compute the maximal $k$-edge connected subgraphs of an undirected simple graph with an improvement in running time for large values of $k$.

In more detail, we provide a data structure for a fully dynamic graph that can be updated in {\em worst-case} time $\hatO(1)$ and that allows (at any point during the sequence of edge updates)
to construct an NMC sparsifier in worst-case time $\hatO(n)$. The probabilistic guarantees work against an {\em adaptive} adversary. If the update time is relaxed to be {\em amortized} or if the adversary is {\em oblivious}, the update time can be improved to $\tilde O(1)$ and the query time to $\tilde O(n)$.

Our basic approach is to maintain suitable data structures in a dynamically changing graph that allow the execution of the \emph{static} NMC sparsifier algorithm based on random 2-out contractions proposed by Ghaffari, Nowicki, and Thorup~\cite{DBLP:conf/soda/Ghaffari0T20} in time $\hatO(n)$ instead of $\tilde O(m)$.
Our main insight is that this speedup can be achieved by maintaining 

\begin{enumerate}
    \item a dynamic spanning forest data structure (DSF) of the input graph $G$ and
    \item a dynamic cut set data structure (DCS), where the user determines which edges belong to a (not necessarily spanning) forest $F$ of $G$, and, given a vertex $v$, the data structure returns an edge that leaves the tree of $F$ that contains $v$ (if it exists).
\end{enumerate}

We show that these two data structures suffice to  construct an NMC sparsifier of the current graph in the desired time bound.
Put differently, we can avoid processing all edges of $G$ in order to build the NMC sparsifier, and only spend time that is roughly proportional to the size of the sparsifier.

Note that the NMC sparsifier is computed from scratch every time it is requested and no information of a previously computed sparsifier is maintained throughout the dynamic updates of the underlying graph. This ensures the probabilistic guarantees hold against an adaptive adversary if the guarantees of the chosen DSF and DCS data structures do.
Our main result is the following theorem.

\begin{restatable}{theorem}{maintheorem}\label{theorem:main}
Let $G$ be a fully dynamic simple graph that currently has $n$ nodes and minimum degree $\delta>0$. There is a data structure that outputs an NMC sparsifier of $G$ that has $O(n/\delta)$ vertices and $O(n)$ edges w.h.p. upon request.
Each update and query takes either
\begin{enumerate}
    \item worst-case $\hatO(1)$ and $\hatO(n)$ time respectively w.h.p., assuming an \emph{adaptive} adversary,\label{theorem:main-worst-case-adaptive}
    or 
    \item amortized $\tildeO(1)$ and $\tildeO(n)$ time respectively w.h.p., assuming an \emph{adaptive} adversary,\label{theorem:main-amortized-adaptive}
    or 
    \item worst-case $\tildeO(1)$ and $\tildeO(n)$ time respectively w.h.p., assuming an \emph{oblivious} adversary.\label{theorem:main-worst-case-oblivious}
\end{enumerate}
\end{restatable}

Recall that $\tildeO$ hides polylogarithmic factors and $\hatO$ hides subpolynomial (i.e., $n^{o(1)}$) factors.
The three cases of Theorem~\ref{theorem:main} result from using different data structures to realize our required DSF and DCS data structures. We show that with minimal overhead both data structures can be reduced to a dynamic minimum spanning forest data structure, for which many constructions have been proposed in the literature.
For Case~\ref{theorem:main-worst-case-adaptive}, we make use of the fully dynamic minimum spanning forest algorithm of Chuzhoy, Gao, Li, Nanongkai, Peng, and Saranurak~\cite{DBLP:conf/focs/ChuzhoyGLNPS20} -- the only spanning forest data structure known so far that can provide worst-case time guarantees against an adaptive adversary.
Case~\ref{theorem:main-amortized-adaptive} is the result of substituting the deterministic data structure of Holm, de Lichtenberg, and Thorup~\cite{DBLP:journals/jacm/HolmLT01} and case~\ref{theorem:main-worst-case-oblivious} results from using  the randomized data structure of Kapron, King, and Mountjoy~\cite{DBLP:conf/soda/KapronKM13} instead.

As a first application  we can efficiently provide a cactus representation of the minimum cuts of a simple graph upon request in the fully dynamic setting. The cactus representation of all minimum cuts of a static graph is known to be computable in near linear time $\tildeO(m)$ using randomized algorithms \cite{DBLP:conf/soda/HeHS24a, DBLP:conf/soda/KargerP09, DBLP:journals/dam/LoST21}. 
With deterministic algorithms, the best known time bound is $m^{1 + o(1)}$ \cite{DBLP:conf/soda/HeHS24a}. We are not aware of any previous work on providing the cactus for fully dynamic graphs in sublinear time per query. Specifically, we show the following:

\begin{restatable}{theorem}{cactustheorem}\label{theorem:cactus}
Let $G$ be a fully dynamic simple graph with $n$ vertices. There is a data structure that w.h.p. gives a cactus representation of all minimum cuts of $G$.
Each update and query takes either
\begin{enumerate}
    \item worst-case $\hatO(1)$ and $\hatO(n)$ time respectively w.h.p., assuming an \emph{adaptive} adversary,\label{theorem:cactus-worst-case-adaptive}
    or 
    \item amortized $\tildeO(1)$ and $\tildeO(n)$ time respectively w.h.p., assuming an \emph{adaptive} adversary,\label{theorem:cactus-amortized-adaptive}
    or 
    \item worst-case $\tildeO(1)$ and $\tildeO(n)$ time respectively w.h.p., assuming an \emph{oblivious} adversary.\label{theorem:cactus-worst-case-oblivious}
\end{enumerate}
\end{restatable}

Theorem~\ref{theorem:cactus} provides an improvement over one of the main technical components in \cite{DBLP:conf/soda/GoranciHNSTW23}. Specifically, Goranci et al.~\cite{DBLP:conf/soda/GoranciHNSTW23}, provide a method to efficiently maintain an NMC sparsifier of a dynamic simple graph, based on the $2$-out contractions of Ghaffari et al.~\cite{DBLP:conf/soda/Ghaffari0T20}. 
However, their main technical result~(Theorem~3.7 in~\cite{DBLP:conf/soda/GoranciHNSTW23}) has several drawbacks. First, it needs to know an estimated upper bound $\hat{\delta}$ on the minimum degree of any graph that occurs during the sequence of updates of $G$.
Second, they try to \emph{maintain} the sparsifier during the updates. This results in an update time $\tilde O(\hat{\delta})$, and forces them to ``hide'' their sparsifier from an adversary, i.e., they can expose their sparsifier only if they work against an oblivious adversary. 
Thus, to make their minimum cut algorithm work against an adaptive adversary, they return only the \emph{value} of the minimum cut.
In contrast, our algorithm computes a sparsifier from scratch upon request, and can therefore provide a cactus representation of \emph{all} minimum cuts, even against an adaptive adversary. 

\clearpage

Notice that we provide a different trade-off in reporting the minimum cut: We have an update time of $\hatO(1)$ and a query time of $\hatO(n)$, whereas Theorem 1.1 of Goranci et al.~\cite{DBLP:conf/soda/GoranciHNSTW23} has an update time of $\tildeO(n)$ and query time $O(1)$. This trade-off is never worse (modulo the subpolynomial factors) if one caches the result of a query and answers queries from the cache if the graph did not change. Furthermore, upon request, we can provide a minimum cut explicitly, and not just its value.

\begin{table}
\renewcommand{\arraystretch}{1.45}\small
\begin{tabular}{ |c|c|c|c| } 
 \hline
 {Algorithm} & {Time} & {Type} & {Range of $k$}\\ \hline\hline
 {Chechik et al., Forster et al.~\cite{DBLP:conf/soda/ChechikHILP17, DBLP:conf/soda/ForsterNYSY20}} & $\tildeO(m+k^{O(k)}n^{3/2})$ & Det. & $k\in\mathbb{N}$\\ \hline
  {Forster et al.~\cite{DBLP:conf/soda/ForsterNYSY20}} & $\tildeO(m+k^3n^{3/2})$ & Las Vegas Rnd. & $k\in\mathbb{N}$ \\ \hline
 {Henzinger et al.~\cite{DBLP:conf/icalp/HenzingerKL15}} & $\tildeO(n^2)$ & Det. & $k\in\mathbb{N}$\\ \hline
 {Thorup, Georgiadis et al.~\cite{DBLP:journals/combinatorica/Thorup07,DBLP:journals/corr/abs-2211-06521}} & $\tildeO(m+k^8n^{3/2})$ & Det. & $k=\log^{O(1)}n$\\ \hline
 {Saranurak and Yuan~\cite{DBLP:conf/esa/SaranurakY23}} & $O(m+kn^{1+o(1)})$ & Det. & $k= \log^{o(1)}{n}$\\ \hline
 {Nalam and Saranurak~\cite{DBLP:conf/soda/NalamS23}} & $\tilde{O}(m\cdot\min\{m^{3/4},n^{4/5}\})$ & Monte Carlo Rnd. & $k\in\mathbb{N}$\\ \hline
 {\textbf{This paper}} & $\tilde{O}(m+{n^2}/k)$ & Monte Carlo Rnd. & $k\in\mathbb{N}$\\ \hline
\end{tabular}

\caption{\small Best known time bounds for computing the maximal $k$-edge-connected subgraphs in undirected graphs in the static setting. The $\tildeO$ expression hides polylogarithmic factors. }\label{table:table1}
\end{table}

As a second application of our main result, we improve the time bounds for computing the (vertex sets of the) maximal $k$-edge-connected subgraphs in a simple undirected graph. 
Specifically, we have the following.

\begin{restatable}{theorem}{maxkedgesubgraphs}\label{theorem:computingmaximal}
Let $G$ be a simple graph with $n$ vertices and $m$ edges, and let $k$ be a positive integer. We can compute the maximal $k$-edge-connected subgraphs of $G$ w.h.p. in $\tildeO(m+{n^2}/k)$ time w.h.p.
\end{restatable}

For comparison, Table~\ref{table:table1} gives an overview of the best known algorithms for computing the maximal $k$-edge-connected subgraphs in undirected graphs.
Thorup~\cite{DBLP:journals/combinatorica/Thorup07} does not deal with the computation of the maximal $k$-edge-connected subgraphs, but the result is a consequence of his fully dynamic minimum cut algorithm as observed in \cite{DBLP:journals/corr/abs-2211-06521}. The algorithm in \cite{DBLP:conf/soda/NalamS23} is the only one that holds for weighted graphs (with integer weights), and it has no dependency on $k$ in the time bound. Notice that our algorithm improves over all prior results for $k =\Omega( n^{1/8})$ and $m =\Omega(n^{8/7})$.

With a reduction to simple graphs, this implies the following bounds for computing the maximal $k$-edge-connected 
subgraphs of multigraphs.

\begin{corollary}
\label{corollary:mult}
Let $G$ be a multigraph with $n$ vertices and $m$ edges, and let $k$ be a positive integer. We can compute the maximal $k$-edge-connected subgraphs of $G$ w.h.p. in $\tildeO(m+k^2n+kn^2)$ time w.h.p.
\end{corollary}

Notice that Corollary~\ref{corollary:mult} provides an improvement compared to the other algorithms in Table~\ref{table:table1}, depending on $k$ and the density of the graph. For example, if $\delta$ and $\epsilon$ are two parameters satisfying $1/4\leq\delta\leq 1$, $16/15\leq\epsilon\leq 2$ and $\delta\leq\epsilon-6/5$, then we have an improvement in the regime where $m=\Theta(n^\epsilon)$ and $k=\Theta(n^\delta)$ against all previous algorithms in Table~\ref{table:table1} that work for multigraphs (i.e., except for Henzinger et al.~\cite{DBLP:conf/icalp/HenzingerKL15}, which works only for simple graphs).

Finally, the method that we use in Theorem~\ref{theorem:computingmaximal} for computing the maximal $k$-edge-connected subgraphs in static graphs extends to the case of dynamic graphs.

\begin{theorem}
\label{theorem:dynamicsubgraphs}
Let $G$ be a fully dynamic simple graph with $n$ vertices. There is a data structure
that can provide the maximal $k$-edge-connected subgraphs of $G$ at any point in time, with high probability, for any integer $k < n$.
Each update and query takes either
\begin{enumerate}
    \item worst-case $\hatO(1)$ and $\hatO(n^2/k)$ time respectively w.h.p, assuming an \emph{adaptive} adversary, or
    \item amortized $\tildeO(1)$ and $\tildeO(n^2/k)$ time respectively w.h.p, assuming an \emph{adaptive} adversary,
    or 
    \item worst-case $\tildeO(1)$ and $\tildeO(n^2/k)$ time respectively w.h.p, assuming an \emph{oblivious} adversary.
\end{enumerate}
\end{theorem}

The results by Aamand~\ea{}~\cite{DBLP:conf/icalp/AamandKLPRT23} and Saranurak and Yuan~\cite{DBLP:conf/esa/SaranurakY23} provide algorithms for maintaining the maximal $k$-edge-connected subgraphs in \emph{decremental} graphs. Georgiadis~\ea{}~\cite{DBLP:journals/corr/abs-2211-06521} provide a fully dynamic algorithm that given two vertices determines whether they belong to the same $k$-edge connected subgraph in time $O(1)$. 
Their worst-case update time is $\tildeO(T(n,k))$, where $T(n,k)$ is the running time of any algorithm for static graphs that is used internally by the algorithm (see values in the time column of Table~\ref{table:table1} with the additive term $m$ removed). Thus, as in the case for dynamic minimum cut, our algorithm provides a different trade-off, with fast updates and slower query time. Notice that Theorem~\ref{theorem:dynamicsubgraphs} provides an improvement over \cite{DBLP:journals/corr/abs-2211-06521} (for any function $T(n,k)$ corresponding to an algorithm from Table~\ref{table:table1}) when $k$ is sufficiently large (i.e., when $k=\Omega(n^{1/8})$).

\subsection{Related Work}
Bencz\'ur and Karger \cite{DBLP:conf/stoc/BenczurK96, DBLP:journals/siamcomp/BenczurK15} demonstrated that \emph{every} cut in a graph can be approximated within a $(1~\pm~\epsilon)$ multiplicative factor using contractions of non-uniformly sampled edges. For a graph on $n$ vertices and $m$ edges, they show how to create a cut sparsifier that contains $O(n \log{n} / \epsilon^2)$ edges in $\tildeO(m)$ time. Fung~\ea{}~\cite{DBLP:journals/siamcomp/FungHHP19} generalized this notion of cut sparsifiers for further non-uniform edge sampling methods, specifically edge strength~\cite{DBLP:conf/stoc/BenczurK96}, effective resistance~\cite{SS11} and standard edge connectivity.
Even stronger is the concept of \emph{spectral} sparsifiers, that approximate the entire Laplacian quadratic form of the input graph to a multiplicative factor $(1 \pm \epsilon)$. This idea was introduced by Spielman and Teng \cite{ST11}, who also showed that every (even weighted) graph admits a spectral sparsifier that contains only $\tildeO(n/\epsilon^2)$ edges and can be computed in near linear time $\tildeO(m/\epsilon^2)$.
Subsequent work by Spielman and Srivastava~\cite{SS11} and later Batson, Spielman, and Srivastava~\cite{BSS12} improved the bounds on the number of edges to $O(n \log{n}/\epsilon^2)$ and $O(n/\epsilon^2)$ edges respectively, where the latter result is optimal up to a constant. Lee and Sun~\cite{LS18} give an almost-linear time construction of a linear-sized spectral sparsifier by sampling according to effective resistance~\cite{SS11}.

While the previously mentioned results aim to approximate \emph{all} cuts of a given graph, it is also possible to construct sparsifiers for undirected graphs that only maintain sufficiently small cuts, but preserve their values exactly.
Let $G$ be a simple undirected graph with $n$ vertices, $m$ edges and minimum degree $\delta$. 
Ghaffari~\ea{}~\cite{DBLP:conf/soda/Ghaffari0T20} designed a Monte Carlo algorithm for contracting $G$ into a graph $\hatG$ in $O(m \log {n})$ time, such that $\hatG$ has $O(n/\delta)$ vertices, $O(n)$ edges and all non-trivial $(2-\epsilon)$ minimum cuts of $G$ are \emph{exactly} preserved. Their method is based on randomly sampling 2 outgoing edges from each vertex and contracting the resulting connected components, giving a \emph{random 2-out contraction}.
For undirected simple graphs, NMC sparsifieres with $\tildeO(n)$ edges and $\tildeO(n/\delta)$ vertices are known to exist since Kawarabayashi and Thorup~\cite{DBLP:journals/jacm/KawarabayashiT19}, who also show how to deterministically construct one in $\tildeO(m)$ time. These results were improved to $O(n)$ edges and $O(n/\delta)$ vertices in $\tildeO(m)$ time by Lo, Schmidt, and Thorup~\cite{DBLP:journals/dam/LoST21}.

For fully dynamic graphs, Abraham~\ea{}~\cite{DBLP:conf/focs/AbrahamDKKP16} show that a $(1 \pm \epsilon)$ approximate cut sparsifier of size $n \cdot \poly(\log{n},1/\epsilon)$ can be maintained in $\poly(\log{n},1/\epsilon)$ worst-case update time.
Bernstein~et~al.~\cite{DBLP:conf/icalp/BernsteinBGNSS022} give an $O(k)$-approximate cut sparsifier of size $\tildeO(n)$ in $\tildeO(n^{1/k})$ amortized update time against an adaptive adversary. They further give a $\polylog(n)$-approximate spectral sparsifier in $\polylog(n)$ amortized update time against an adaptive adversary.
Goranci~\ea{}~\cite{DBLP:conf/soda/GoranciHNSTW23} are able to maintain an NMC sparsifier of a fully dynamic graph against an oblivious adversary, utilizing an efficient dynamic expander decomposition.
In particular, they show how to maintain an NMC sparsifier with worst-case $\tildeO(n)$ update time that can answer queries for the value of the minimum cut w.h.p. in $O(1)$ time.

The cactus representation of all minimum cuts of a static graph is known to be computable in near linear time $\tildeO(m)$ using randomized algorithms \cite{DBLP:conf/soda/HeHS24a, DBLP:conf/soda/KargerP09, DBLP:journals/dam/LoST21}. The fastest such algorithm achieves a running time of $O(m \log^3{n})$ \cite{DBLP:conf/soda/HeHS24a}. With deterministic algorithms, the best known time bound is $m^{1 + o(1)}$ \cite{DBLP:conf/soda/HeHS24a}. We are not aware of any previous work on providing the cactus for fully dynamic graphs in sublinear time per query.

\section{Preliminaries}
\label{sec:preliminaries}
In this paper, we consider only undirected, unweighted graphs. 
A graph is called \emph{simple} if it contains no parallel edges, conversely it is a multigraph if it does.
We use common graph terminology and notation that can be found e.g. in \cite{DBLP:books/cu/NI2008}. Let $G=(V,E)$ be a graph. Throughout, we use $n$ and $m$ to denote the number of vertices and edges of $G$, respectively. We use $V(G)$ and $E(G)$ to denote the set of vertices and edges, respectively, of $G$; that is, $V=V(G)$ and $E=E(G)$. A \emph{subgraph} of $G$ is a graph of the form $(V',E')$, where $V'\subseteq V$ and $E'\subseteq E$. A \emph{spanning} subgraph of $G$ is a subgraph of the form $(V,E')$ that contains at least one incident edge to every vertex of $G$ that has degree $>0$. If $X$ is a subset of vertices of $G$, we let $G[X]$ denote the induced subgraph of $G$ on the vertex set $X$. Thus, $V(G[X]):=X$, and the edge set of $G[X]$ is $\{e\in E\mid \mbox{ both endpoints of } e \mbox{ are in } X\}$. If $E'$ is a set of edges of $G$, we let $G[E']:=(V,E')$.

A connected component of $G$ is a maximal connected subgraph of $G$. A set of edges of $G$ whose removal increases the number of connected components of $G$ is called a \emph{cut} of $G$. The size $|C|$ is called the \emph{value} of the cut. A cut $C$ with minimum value in a connected graph $G$ is called a \emph{minimum cut} of $G$. In this case, $G\setminus C=(V,E\setminus C)$ consists of two connected components. If one of them is a single vertex, then $C$ is called a \emph{trivial minimum cut}.

An \emph{NMC sparsifier} of $G$ is a multigraph $H$ on the same vertex set as $G$ that preserves all non-trivial minimum cuts of $G$, in the sense that the number of edges leaving any non-trivial minimum cut is the same in $H$ as it is in $G$.

\compactpara{Dynamic graphs.}
A \emph{dynamic} graph is a graph that changes over time. Thus, it can be thought of as a sequence of graphs $G_1,\dots,G_t$, where every $G_i$ differs from $G_{i-1}$ by one element (i.e., it has one more/less edge/vertex), and the indices $i\in\{1,\dots,t\}$ correspond to an increasing sequence of discrete time points. Therefore, at any point in time we may speak of the \emph{current} graph. Accordingly, an algorithm that handles a dynamic graph relies on a data structure representation of the current graph, which is updated with commands of the form $\mathtt{insert}$ and $\mathtt{delete}$, corresponding to the changes to the graph each time.
If the graph only increases or only decreases in size, it is called \emph{partially} dynamic. If \emph{any} mixture of changes is allowed, it is called \emph{fully} dynamic. In this paper, we consider fully dynamic graphs.

\compactpara{Contractions of graphs.}
Let $E'\subseteq E$ be a set of edges of a graph $G=(V,E)$, and let $C_1,\dots,C_t$ be the connected components of the graph $G'=(V,E')$. The \emph{contraction} $\hatG$ induced by $E'$ is the graph derived from $G$ by contracting each $C_i$, for $i\in\{1,\dots,t\}$, into a single node. We ignore possible self-loops, but we maintain distinct edges of $G$ that have their endpoints in different connected components of $G'$.  Hence, $\hatG$ may be a multigraph even though $G$ is a simple graph. Furthermore, there is a natural injective correspondence from the edges of $\hatG$ to the edges of $G$. We say that an edge of $G$ is \emph{preserved} in $\hatG$ if its endpoints are in different connected components of $G'$ (it corresponds to an edge of $\hatG$). 

A \emph{random $2$-out contraction} of $G$ is a contraction of $G$ that is induced by sampling from every vertex $v$ of $G$ two edges incident to it (independently, with repetition allowed) and setting $E'$ to be the edges sampled in this way. Thus,  $E'$  satisfies $|E'|\leq 2|V(G)|$. The importance of considering $2$-out contractions is demonstrated in the following theorem of Ghaffari et al.~\cite{DBLP:conf/soda/Ghaffari0T20}.

\begin{theorem}[rephrased weaker version of Theorem 2.3 in \cite{DBLP:conf/soda/Ghaffari0T20}]
\label{theorem:2outcontraction}
A random $2$-out contraction of a simple graph $G$ with $n$ vertices and minimum degree $\delta$ has $O(n/\delta)$ vertices, with high probability, and preserves any fixed non-trivial minimum cut of $G$ with constant probability. 
\end{theorem}

\compactpara{Preserving small cuts via forest decompositions.}
Nagamochi and Ibaraki~\cite{DBLP:journals/algorithmica/NagamochiI92} have shown the existence of a sparse subgraph that maintains all cuts of the original graph with value up to $k$, where $k$ is an integer $\geq 1$. Specifically, given a graph $G=(V,E)$ with $n$ vertices and $m$ edges, there is a spanning subgraph $H=(V,E')$ of $G$ with $|E'|\leq k(n-1)$ such that every set of edges $C\subseteq E$ with $|C|< k$ is a cut of $G$ if and only if $C$ is a cut of $H$. A graph $H$ with this property is given by a \emph{$k$-forest decomposition} of $G$, which is defined as follows. First, we let $F_1$ be a spanning forest of $G$. Then, for every $i\in\{2,\dots,k\}$, we recursively let $F_i$ be a spanning forest of $G\setminus(F_1\cup\dots\cup F_{i-1})$. Then we simply take the union of the forests $H=F_1\cup\dots\cup F_k$. A naive implementation of this idea takes time $O(k(n+m))$. However, this construction can be completed in linear time by computing a maximum adjacency ordering of $G$ \cite{DBLP:books/cu/NI2008}.

\compactpara{Maximal \mathversion{bold}$k$\mathversion{normal}-edge-connected subgraphs.}
Given a graph $G$ and a positive integer $k$, a \emph{maximal $k$-edge-connected subgraph} of $G$ is a subgraph of the form $G[S]$, where: $(1)$ $G[S]$ is connected, $(2)$ the minimum cuts of $G[S]$ have value at least $k$, and $(3)$ $S$ is maximal with this property. The first two properties can be summarized by saying that $G[S]$ is $k$-edge-connected, which equivalently means that we have to remove at least $k$ edges in order to disconnect $G[S]$. It is easy to see that the vertex sets $S_1,\dots,S_t$ that induce the maximal $k$-edge-connected subgraphs of $G$ form a partition of $V$. 

\compactpara{A note on probabilistic guarantees.}
Throughout the paper we use the abbreviation w.h.p.\  (with high probability) to mean a probability of success at least $1-O(\frac{1}{n^c})$, where $c$ is a fixed constant chosen by the user and $n$ denotes the number of vertices
of the graph. 
The choice of $c$ only affects the values of two parameters $q$ and $r$ that are used internally by the algorithm of Ghaffari et al.~\cite{DBLP:conf/soda/Ghaffari0T20}, which our result is based on. 
Note that the specification ``w.h.p'' may be applied either to the running time or to the correctness (or to both).

\section{Outline of Our Approach}
\label{sec:review}
Our main contribution is a new data structure that outputs an NMC sparsifier of a fully dynamic graph upon request. The idea is to compute the NMC sparsifier from scratch every time it is requested. For this we adapt the construction by Ghaffari et al.~\cite{DBLP:conf/soda/Ghaffari0T20} to compute a random $2$-out contraction of the graph at the current time. We can achieve a speedup of the construction time by maintaining just two data structures for dynamic forests throughout the updates of the graph.

\subsection{Updates: Data Structures for Dynamic Forests} \label{sec:dynamicforests}
As our fully dynamic graph $G$ changes, we rely on efficient data structures for the following two problems, which we call the ``dynamic spanning forest`` problem (DSF) and the ``dynamic cutset'' problem (DCS).
In DSF, the goal is to maintain a spanning forest $F$ of a dynamic graph $G$. Specifically, the DSF data structure supports the following operations.

\begin{itemize}
\item{$\mathtt{insert}(e)$. Inserts a new edge $e$ to $G$. If $e$ has endpoints in different trees of $F$, then it is automatically included in $F$, and this event is reported to the user.}
\item{$\mathtt{delete}(e)$. Deletes the edge $e$ from $G$. If $e$ was a tree edge in $F$, then it is also removed from $F$, and a new replacement edge is selected among the remaining edges of $G$ in order to reconnect the trees of $F$ that got disconnected. If a replacement edge is found, it automatically becomes a tree edge in $F$, and it is output to the user.} 
\end{itemize}

\noindent
The DCS problem is a more demanding variant of DSF. Here the goal is to maintain a (not necessarily spanning) forest $F$ of the graph, but we must also be able to provide an edge that connects two different trees if needed. Specifically, the DCS data structure supports the following operations.

\begin{itemize}
\item{$\mathtt{insert}_G(e)$.} Inserts a new edge $e$ to $G$.
\item{$\mathtt{insert}_F(e)$.} Inserts an edge $e$ to $F$, if $e$ is an edge of $G$ that has endpoints in different trees of $F$. Otherwise, $F$ stays the same.
\item{$\mathtt{delete}_G(e)$.} Deletes the edge $e$ from $G$. If $e$ is a tree edge in $F$, it is also deleted from $F$.
\item{$\mathtt{delete}_F(e)$.} Deletes the edge $e$ from $F$ (or reports that $e$ is not an edge of $F$).
\item{$\mathtt{find\_cutedge}(v)$.}
 Returns an edge of $G$ (if it exists) that has one endpoint in the tree of $F$ that contains $v$, and the other endpoint outside of that tree.
\end{itemize}

\noindent
The main difference between DSF and DCS is that in DCS the user has absolute control over the edges of the forest $F$. In both problems, the most challenging part is finding an edge that connects two distinct trees of $F$. In DCS this becomes more difficult because the data structure must work with the forest that the user has created, whereas in DSF the spanning forest is managed internally by the data structure.
Both of these problems can be solved by reducing them to the dynamic minimum spanning forest (MSF) problem, as shown in the following Lemma~\ref{lemma:reducetoMSF}, which is proven in the Appendix~\ref{sec:appendixproof}.

\begin{lemma}
\label{lemma:reducetoMSF}
The DSF problem can be reduced to the DCS problem within the same time bounds. The DCS problem can be reduced to dynamic MSF with an additive $O(\log n)$ overhead for every operation. 
\end{lemma}

To realize these two data structures deterministically with worst-case time guarantees, the only available solution at the moment is to make use of the reduction to the dynamic MSF problem given in Lemma~\ref{lemma:reducetoMSF} and then employ the dynamic MSF data structure of Chuzhoy~et.~al.~\cite{DBLP:conf/focs/ChuzhoyGLNPS20}, which supports every update operation in worst-case $\hatO(1)$ time. 
Alternatively, one can solve both DSF and DCS deterministically with \emph{amortized} time guarantees using the data structures of Holm~et.~al.~\cite{DBLP:journals/jacm/HolmLT01}. In that case, every update for DSF can be performed in amortized $O(\log^2{n})$ time, and every operation for DCS can be performed in amortized $O(\log^4{n})$ time, by reduction to dynamic MSF.  
If one is willing to allow randomness, one can solve both DSF and DCS in worst-case polylogarithmic time per operation, under the assumption of an oblivious adversary with the data structure by Kapron~et~al.~\cite{DBLP:conf/soda/KapronKM13}.

These different choices for realizing the dynamic MSF data structure give the following lemma.

\begin{lemma}\label{lemma:dyn-forest-datastructures}
    The DSF and DCS data structures can be realized in either
    \begin{enumerate}
    \item deterministic worst-case $\hatO(1)$ update time,\label{lemma:dyn-forest-worst-case-adaptive} 
    or 
    \item deterministic amortized $\tildeO(1)$ update time,\label{lemma:dyn-forest-amortized-adaptive} 
    or 
    \item worst-case $\tildeO(1)$ update time, assuming an \emph{oblivious} adversary.\label{lemma:dyn-forest-worst-case-oblivious}
\end{enumerate}
\end{lemma}

To handle the updates on $G$, any insertion or deletion of an edge is also applied in both the DSF and the DCS data structure. In particular, for DCS we only use the $\mathtt{insert}_G$ and $\mathtt{delete}_G$ operations and keep its forest empty.
At query time, we will build the DCS forest from scratch and then fully delete it again. In order for this to be efficient, DCS needs to have already processed all edges of $G$.

\subsection{Queries: Efficiently Contracting a Dynamic Graph}\label{sec:outline-queries}

The NMC sparsifier that we output for each query is a random 2-out contraction of the graph at the current time. For this construction, we use our maintained forest data structures and build upon the algorithm of Ghaffari et al.~\cite{DBLP:conf/soda/Ghaffari0T20}, who prove the following. 

\begin{theorem}[Weaker version of Theorem 2.1 of Ghaffari et al.~\cite{DBLP:conf/soda/Ghaffari0T20}]
\label{theorem:contractedgraph}
Let $G$ be a simple graph with $n$ vertices, $m$ edges, and minimum degree $\delta$. In $O(m\log n)$ time we can create a contraction $\hatG$ of $G$ that has $O(n/\delta)$ vertices and $O(n)$ edges w.h.p., and preserves all non-trivial minimum cuts of $G$ w.h.p.
\end{theorem}

Note that in particular, the contracted graph $\hatG$ created by the above theorem is an NMC sparsifier, as desired.
We first sketch the algorithm that yields Theorem~\ref{theorem:contractedgraph} as described by Ghaffari et al.~\cite{DBLP:conf/soda/Ghaffari0T20}. Then we outline how we can adapt and improve this construction for dynamic graphs, making use of the DSF and DCS data structures. Specifically, to answer a query we show that we can build $\hatG$ in time proportional to the size of $\hatG$, which may be much lower than $O(m\log n)$. The details and a formal analysis can be found in Section~\ref{section:main}.

\paragraph*{Random 2-out Contractions of Static Graphs.}
Ghaffari et al.'s algorithm~\cite{DBLP:conf/soda/Ghaffari0T20} works as follows (see also Algorithm~\ref{algorithm:contraction}).
First, it creates a collection of $q=O(\log n)$ random $2$-out contractions $G_1,\dots,G_q$ of $G$, where every $G_i$ is created with independent random variables.
Now, according to Theorem 2.4 in \cite{DBLP:conf/soda/Ghaffari0T20}, each $G_i$ for $i\in\{1,\dots,q\}$,  has $O(n/\delta)$ vertices w.h.p., and preserves any fixed non-trivial minimum cut of $G$ with constant probability. 

In a second step, they compute a $(\delta+1)$-forest decomposition $\hatG_i$ of $G_i$, for every $i\in\{1,\dots,q\}$, in order to ensure that $\hatG_i$ has $O(\delta\cdot( n/\delta))=O(n)$ edges w.h.p.
Because $G$ has minimum degree $\delta$, every non-trivial minimum cut has value at most $\delta$. Hence, each $\hatG_i$ still maintains every fixed non-trivial minimum cut with constant probability.

\begin{algorithm}[t]
\caption{\textsf{Construction of the contracted graph $\hatG$ in Theorem~\ref{theorem:contractedgraph} }}
\label{algorithm:contraction}
\LinesNumbered
\DontPrintSemicolon
\textbf{input}: a simple graph $G$ in adjacency list representation\;
choose parameters $q,r=\Theta(\log n)$ according to \cite{DBLP:conf/soda/Ghaffari0T20}\;
compute the minimum degree $\delta$ of $G$\;
\For{$i\leftarrow 1$ to $q$}{
  construct a $2$-out contraction $G_i$ of $G$\;
}
\ForEach{$i\in\{1,\dots,q\}$}{
  construct a $(\delta+1)$-forest decomposition $\hatG_i$ of $G_i$\;
}
let $\Econ\leftarrow\emptyset$ \textcolor{cblue}{\tcp{a set of edges of $G$ to contract}}
\ForEach{edge $e$ of $G$}{
  \If{$e$ is preserved in less than $r$ graphs from $\hatG_1,\dots,\hatG_q$}{
    $\Econ\leftarrow \Econ\cup\{e\}$\;
  }
}
\textbf{return} the graph obtained from $G$ by contracting the edges in $\Econ$\;
\end{algorithm}

Finally, they select a subset of edges $\Econ\subseteq E(G)$ to contract by a careful \emph{``voting''} process. Specifically, for every edge $e$ of $G$, they check if it is an edge of at least $r$ graphs from $\hatG_1,\dots,\hatG_q$, where $r$ is a carefully chosen parameter.
If $e$ does not satisfy this property, then it is included in $\Econ$, the set of edges to be contracted. In the end, $\hatG$ is given by contracting all edges from $\Econ$.

\paragraph*{An Improved Construction using Dynamic Forests.}
We now give an overview of our algorithm for efficiently constructing the NMC sparsifier $\hatG$. It crucially relies on having the DSF and DCS data structures initialized to contain the edges of the current graph $G$. For a fully dynamic graph, this is naturally ensured at query time by maintaining the two forest data structures throughout the updates, as described in Section~\ref{sec:dynamicforests}.
Since the goal is to output a graph $\hatG$ with $O(n/\delta)$ vertices and $O(n)$ edges w.h.p., we aim for an 
algorithm that takes roughly $O(n)$ time. The process is broken up into three parts. 

\compactpara{Part 1.} First, we compute the $2$-out contractions $G_1,\dots,G_q$ for $q=O(\log n)$. Each $G_i$ can be computed by sampling, for every vertex $v$ of $G$, two random edges incident to $v$ (independently, with repetition allowed). Since the graph $G$ is dynamic, the adjacency lists of the vertices also change dynamically, and therefore this is not an entirely trivial problem. However, in Appendix~\ref{sec:sampling} we provide an efficient solution that relies only on elementary data structures. 
Specifically, we can support every graph update in worst-case constant time, and every query for a random incident edge to a vertex in worst-case $\tilde{O}(1)$ time. Notice that each $G_i$, for $i\in\{1,\dots,q\}$, is given by contracting a set $E_i$ of $O(n)$ edges of $G$. 

\compactpara{Part 2.} Next, we separately compute a $(\delta+1)$-forest decomposition $\hatG_i$ of $G_i$, for every $i\in\{1,\dots,q\}$. 
Each $\hatG_i$ is computed by only accessing the edges of $E_i$ plus the edges of the output (which are $O(n)$ w.h.p.). 
For this, we rely on the DCS data structure. 
Since in DCS we have complete control over the maintained forest $F$, we can construct it in such a way, that every connected component of $G[E_i]$ induces a subtree of $F$. 
Notice that the connected components of $G[E_i]$ correspond to the vertices of $G_i$. 
This process of modifying $F$ takes $O(|E_i|)=O(n)$ update operations on the DCS data structure. Then, we have established the property that the tree edges that connect vertices of different connected components of $G[E_i]$ correspond to a spanning tree of $G_i$. 
Afterwards, we just repeatedly remove the spanning tree edges, and find replacements using the DCS data structure. These replacements constitute a new spanning forest of $G_i$. 
Thus, we just have to repeat this process $\delta+1$ times to construct the desired $(\delta+1)$-forest decomposition.
Note that every graph $\hatG_i$ is constructed in time roughly proportional to its size (which is $O(n)$ w.h.p.).
Any overhead comes from the use of the DCS data structure.

\compactpara{Part 3.}
Finally, we construct the graph $\hatG$ by contracting the edges $E_{<r}$ of $G$ that appear in less than $r$ graphs from $\hatG_1,\dots,\hatG_q$. From Ghaffari et al.~(\cite{DBLP:conf/soda/Ghaffari0T20}, Theorem 2.4) it follows that $|E\setminus E_{<r}| =|E_{\ge r}|=O(qn)=O(n\log n)$. In the following we provide an algorithm for constructing $\hatG$ with only $O(n + |E_{\ge r}|)$ operations of the DSF data structure.

We now rely on the spanning forest $F$ of $G$ that is maintained by the DSF data structure.
We pick the edges of $F$ one by one, and check for every edge whether it is contained in $E_{<r}$ (it is easy to check for membership in $E_{<r}$ in $O(\log n)$ time).
If it is not contained in $E_{<r}$, then we store it as a \emph{candidate edge}, i.e., 
an edge that possibly is in $\hatG$. In this case, we also remove it from $F$, and the DSF data structure will attempt to fix the spanning forest by finding a replacement edge.
Otherwise, if $e\in E_{<r}$, then we ``fix'' it in $F$, and do not process it again.

In the end all ``fixed'' edges of $F$ form a spanning forest of the
connected components of $G[E_{<r}]$.
Note that the algorithm makes continuous progress: in each step it either identifies an edge of $E_{\ge r}$, or it ``fixes'' an edge of $F$ (which can happen at most $n-1$ times).
 Thus, it performs $O(n+|E_{\ge r}|)=\tildeO(n)$ DSF operations. Since
 we have arrived at a spanning forest of $G[E_{<r}]$, we can build $\hatG$ by contracting the candidate edges stored during this process. 

\section{Analysis}
\label{section:main}

Our main technical tools are Propositions~\ref{proposition:kforests} and \ref{proposition:contract}. Proposition~\ref{proposition:kforests} allows us to create a forest decomposition of a contracted graph without accessing all the edges of the original graph, but roughly only those that are included in the output, and those that we had to contract in order to get the contracted graph. Proposition~\ref{proposition:contract} has a somewhat reverse guarantee: it allows us to create a contracted graph by accessing only the edges that are preserved during the contraction (plus only a small number of additional edges).

In this section we assume that $G$ is a fully dynamic graph that currently has $n$ vertices. We also assume that we have maintained a DCS and a DSF data structure on $G$, which support every operation in $\UCS$ and $\USF$ time, respectively (cf. Section~\ref{sec:dynamicforests}).

\subsection{A k-Forest-Decomposition of the Contraction}
Given a set of edges to contract, the following proposition shows how to compute a $k$ forest decomposition of the induced contracted graph. Crucially, the contracted graph $H$ does not have to be given explicitly and does not need to be computed. Instead, the proposition shows that it is sufficient to know only a set of edges whose contraction yields $H$. Thus, the running time of the construction is independent of the number of edges of $H$. The whole procedure is shown in Algorithm~\ref{algorithm:kforests}.

\begin{restatable}{proposition}
{kforests}
\label{proposition:kforests}
Let $H$ be a contraction of $G$ that results from contracting a given set $\Econ$ of edges in $G$. Let further $k$ be a positive integer and $n_H$ be the number of vertices in $H$.
Then we can construct a $k$-forest decomposition of $H$ in time $O\big((n+kn_H)\cdot\UCS+|\Econ|\log{n}\big)$.
\end{restatable}

\begin{algorithm}[t]
\caption{\textsf{Compute a $k$-forest decomposition of the graph that is formed by contracting a set of edges $E'$ of $G$}}
\label{algorithm:kforests}
\LinesNumbered
\DontPrintSemicolon
let $F$ be the empty DCS forest\;
\ForEach{edge $e\in E'$}{
  \If{the endpoints of $e$ belong to different trees of $F$}{
     $\mathtt{DCS.insert}_F(e)$\;
  }
}
let $\mathcal{V}$ be a set that consists of one vertex from every tree of $F$\;
set $S\leftarrow\emptyset$\;
\For{$i\leftarrow 1$ to $k$}{
\label{algorithm:kforests:line1}
  set $L\leftarrow\emptyset$ \textcolor{cblue}{\tcp{$L$ will contain the edges of the current spanning forest}}
  \ForEach{$v\in\mathcal{V}$}{
    let $e\leftarrow\mathtt{DCS.find\_cutedge}(v)$\;
    \While{$e\neq\bot$}{
    \label{algorithm:kforests:while}
      $\mathtt{DCS.insert}_F(e)$, and append $e$ to $L$\;
      $e\leftarrow\mathtt{DCS.find\_cutedge}(v)$\;
    }
  }
  \ForEach{$e\in L$}{
    $\mathtt{DCS.delete}_G(e)$, and append $e$ to $S$\;
  }
}
\label{algorithm:kforests:line2}

use $\mathtt{DCS.delete}_F$ to remove all the edges from $F$\;
\ForEach{$e\in S$}{
  $\mathtt{DCS.insert}_G(e)$ \textcolor{cblue}{\tcp{restore DCS to its original state}}
}
\textbf{return} $S$\;
\end{algorithm}

\noindent Note that the number of DCS operations depends only on (1) the number $n$ of vertices of $G$, (2) the number $n_H$ of vertices of $H$, and (3) the number $k$ of forests that we want to build.  

\begin{proof}
We assume the DCS data structure contains all edges of the current graph $G$ and an empty forest $F$.
We first want to construct $F$ as a spanning forest of $G[\Econ]$.
To do this, we process all edges $e \in \Econ$ one by one, checking for each edge $e$ if its endpoints belong to the same tree of $F$. 
If yes, we do nothing. Otherwise, we insert $e$ into $F$ as a tree edge. To perform this check efficiently, we use a disjoint-set union (DSU) data structure that supports a sequence of $\ell$ union operations in $O(\ell\log n)$ total time, and every query in constant time (see e.g.~\cite{DBLP:books/daglib/0023376}).
Now $F$ is a spanning forest of $G[\Econ]$. Note that this took $O(n\UCS)$ time for the DCS operations plus $O(|\Econ|\log n)$ for the DSU operations. 

Next we compute the connected components of $F$ and choose a representative vertex from each component. This takes $O(n)$ time (by processing all trees of $F$). Note that there are $n_H$ components as there is a one-to-one correspondence between the components of $F$ and the vertices of $H$.
Now, for each representative $v$, we repeatedly call the operation $\mathtt{find\_cutedge}(v)$ of the DCS data structure to get an edge with exactly one endpoint in the tree of $F$ that contains $v$. As long as such an edge is found, we insert it into $F$ and into a list $L$ and repeat. Then we proceed with the next representative.

In this way, we have built a spanning forest of $H$, which consists of the edges in $L$ by making $O(n_H)$ calls to the DCS data structure. Then we remove all edges in $L$ from the graph in the DCS data structure using $\mathtt{delete}_G$, store them in a list $S$, and repeat the same process $k$ times. In the end, the edges in $S$ form the desired $k$-forest decomposition of $H$.

Finally, we re-insert the edges from $S$ into the DCS data structure using $\mathtt{insert}_G$ and delete the forest $F$ using $\mathtt{delete}_F$ in order to restore its original state. 
Note that $S$ and $F$ contain $O(kn_H)$ and $O(n)$ edges respectively.
This algorithm thus runs in $O\big((n+kn_H)\cdot\UCS+|\Econ|\log{n}\big)$ time.
\end{proof}

\subsection{Building the Contracted Graph}

A contracted graph can naturally be computed from its defining set of contraction edges $\Econ$ in time proportional to the size of this set $|\Econ|$.
Recall that in our case $\Econ$ is the result of the ``voting'' procedure across all generated $\delta$-forest decompositions (cf. Section~\ref{sec:outline-queries}), which is rather expensive to compute.
We hence use a different construction that does not need to know $\Econ$ explicitly. Instead, it relies on an efficient oracle to certify that a given edge is not contained in $\Econ$.

\begin{restatable}{proposition}
{propositioncontract}
\label{proposition:contract}
Let $H$ be a contraction of $G$ that results from contracting a set $\Econ$ of edges in $G$ and let $\Epre$ be the set of edges of $G$ that are preserved in $H$. Suppose that there
is a set of edges $\Echeck$ with $\Epre\subseteq \Echeck$ and $\Echeck\cap\Econ=\emptyset$, for which 
we can check membership in time $\mu$.  
Then we can construct $H$ in time $O(n\mu+|\Echeck|\cdot(\USF+\mu))$.
\end{restatable}

\noindent
Note that the number of DSF operations is proportional to the number of edges in $\Echeck$ (the set $\Econ$ is not required as input to the algorithm).
Thus, this algorithm becomes more efficient if only few edges are contained in $E(G)\setminus \Econ$ (since $\Echeck\subseteq E(G)\setminus\Econ$). In our application, we will have $\mu=O(\log n)$. 

\begin{proof}
Let $F$ be the spanning forest of $G$ that is maintained by the DSF data structure. We start 
by putting the edges of $F$ on a stack $S$. While $S$ is not empty, we pop an edge $e$ from $S$, and check whether $e$ is in $\Echeck$. If $e\notin\Echeck$, then we do nothing (i.e., we keep $e$ on $F$).
Otherwise we store $e$ in a list $L$ of candidate edges (edges that may be preserved in $H$), and call $\mathtt{delete}(e)$ to remove it from the DSF data structure. If this deletion returns a replacement edge $e'$ for $e$, we put $e'$ on the stack. 

After these deletions, $F$ is a spanning forest for the connected components of $G[\Econ]$. This means, there are no more edges of $G$ that connect different connected components of $G[\Econ]$. Consequently, the list $L$ contains all edges of $\Epre$ (and maybe some extra edges). 

To create $H$, we first assign labels to the vertices of 
$G$, so that a label at a vertex $v$ identifies the connected
component of $G[\Econ]$ that $v$ belongs to. This can be done in time $O(n)$ by a graph traversal on the edges of $F$. Note that there is a one-to-one correspondence between the
connected components of $G[\Econ]$ and the vertices of $H$. Thus, we may use the labels for the connected components of $G[\Econ]$ as the vertex set of $H$.
Now we process the edges of $L$ one by one. If an edge $e\in L$ has both endpoints in the same connected component of $G[\Econ]$ we ignore it. Otherwise, we create a new edge in $H$ between
the labels of the endpoints of $e$.
Thus we have constructed $H$. Finally, we re-insert all edges from $L$ into the DSF data structure.

Now we argue about the running time of this procedure. In the first phase
when we process the stack $S$, we continuously make progress: either we identify an edge that will belong to the final spanning forest, or we process an edge that belongs to $\Echeck$. Hence, we process $O(n+|\Echeck|)$ edges in this phase, and for each of those we have to check for membership in $\Echeck$. This takes $O((n+|\Echeck|)\mu)$ time.
Every edge from $\Echeck$ that we process has to be deleted from the DSF data structure. This gives
an additional term of $O(|\Echeck|\cdot\USF)$.
Then, the computation of the vertex sets of the connected components of $G[\Econ]$ takes time $O(n)$. Finally, the construction of $H$ takes time $O(n+|\Echeck|)$. Thus, the total running time is at most $O(n\mu+|\Echeck|\cdot(\USF+\mu))$. 
\end{proof}

\subsection{Constructing an NMC sparsifier}

We can now state our result for computing an NMC sparsifier of a simple graph using Propositions~\ref{proposition:kforests} and \ref{proposition:contract}. 
Compare this to Theorem~\ref{theorem:contractedgraph} and note how Theorem~\ref{theorem:constructing-sparsifier} requires initialized DSF and DCS data structures but has a running time that is independent of the number of edges in $G$.
An outline of this procedure is given in Section~\ref{sec:outline-queries}.

\begin{restatable}{theorem}{theoremsparsifier}\label{theorem:constructing-sparsifier}
Let $G$ be a simple graph with $n$ vertices that has minimum degree $\delta>0$ and is maintained in a DSF and a DCS data structure. 
Then, with high probability we can construct an NMC sparsifier of $G$ that has $O(n/\delta)$ vertices and $O(n)$ edges.
W.h.p this takes $\tildeO(n\cdot(\USF+\UCS))$ time.
\end{restatable}

\begin{proof}
We begin by computing $q=O(\log{n})$ $2$-out contractions $G_1,\dots,G_q$ of $G$. For each contraction, we need to sample two random edges incident to every node independently and with repetition allowed. Since the adjacency lists of the vertices change dynamically, we use the data structure described in Appendix~\ref{sec:sampling}, which supports sampling a random element from a dynamic list in worst-case $O(\log n)$ time. Thus, the edge-sets that induce the $2$-out contractions $G_1,\dots,G_q$ can be sampled in $\tildeO(n)$ total time.

According to Theorem~\ref{theorem:2outcontraction}, every $2$-out contraction $G_i\in\{G_1,\dots,G_q\}$ has $O(n/\delta)$ vertices w.h.p. We use the algorithm described in Proposition~\ref{proposition:kforests} to construct a $(\delta+1)$-forest decomposition $\hatG_i$ of $G_i$, for every $i\in\{1,\dots,q\}$.
Note that $\delta$ can be computed in time $O(n)$ by traversing all vertices of $G$.
Since every $G_i$ has $O(n/\delta)$ vertices w.h.p. and is the result of contracting $O(n)$ edges, this takes time $O\big((n+\delta\cdot O(n/\delta))\UCS+O(n)\log{n})=\tildeO(n\cdot\UCS\big)$ w.h.p.\ for each $G_i$, and, hence, $\tildeO(n\cdot\UCS)$ w.h.p. in total.

In order to get the final contraction $\hatG$, we have to perform the ``voting'' process that selects the edges $\Econ$ of $G$ to contract. Recall that an edge of $G$ belongs to $\Econ$ iff it is preserved in less than $r$ graphs $\hatG_1,\dots,\hatG_q$. Thus, we apply Proposition~\ref{proposition:contract} with 
$\Econ=E_{<r}$ and $\Echeck=E_{\ge r}$, where $E_{\ge r}$ (resp., $E_{<r}$) is the set of edges that are preserved by at least (resp., strictly less than) $r$ graphs from $\{\hatG_1,\dots,\hatG_q\}$.
Observe that this choice fulfills the requirements of Proposition~\ref{proposition:contract}.

It remains to show how we can check for membership in $\Echeck=E_{\ge r}$. For
this we insert every edge from $E(\hatG_1)\cup\dots\cup E(\hatG_q)$ into a balanced binary search tree (BST),
and maintain a counter of how many times it occurs in one of 
the graphs $\hatG_1,\dots,\hatG_q$. Then, for a membership query we can just check in $O(\log n)$ time if the counter-value is at least $r$ (an edge that is not in the BST implicitly has a counter of $0$).

Since the total number of edges in all graphs $\hatG_1,\dots,\hatG_q$ is $\tildeO(n)$ w.h.p., we have that the number of edges in $\Echeck$ is at most $\tildeO(n)$ w.h.p. Thus, Proposition~\ref{proposition:contract} implies that $\hatG$ can be constructed in time $O(n\log{n}+\tildeO(n)(\USF+\log{n}))=\tildeO(n\cdot\USF)$ time w.h.p. 

In total, the construction of the NMC sparsifier takes $\tildeO(n\cdot(\USF+\UCS))$ time w.h.p.
\end{proof}

\subsection{Fully Dynamic Graphs}

The result of Theorem~\ref{theorem:constructing-sparsifier} can easily be extended to fully dynamic simple graphs by maintaining the DSF and DCS data structures throughout the updates of the graph. These forest data structures can be realized in different ways, as described in Section~\ref{sec:dynamicforests}. Depending on this choice we get a different result, and this is how we derive Theorem~\ref{theorem:main}.

\maintheorem*
\begin{proof}
    Each update to $G$ initiates an update to the DSF and DCS data structures, and also to the data structure for sampling from dynamic lists described in Appendix~\ref{sec:sampling}. Since the latter supports every update in worst-case constant time, we conclude that every update on $G$ is processed in $O(\USF+\UCS)$ time.
    Since the forest data structures are properly maintained, we can use Theorem~\ref{theorem:constructing-sparsifier} to answer queries for an NMC sparsifier in $\tildeO(n\cdot(\USF+\UCS))$ time w.h.p. The update times for the forest data structures can be chosen according to Lemma~\ref{lemma:dyn-forest-datastructures}.
\end{proof}
If $G$ is disconnected, all minimum cuts have value $0$ and an NMC sparsifier $H$ is by definition only required to have no edges between different connected components of $G$. Crucially, this implies that there is no guarantee that any information of the minimum cuts within each connected component of $G$ is preserved in $H$.
In this case, however, we can easily strengthen the result by applying Theorem~\ref{theorem:constructing-sparsifier} to each connected component of $G$ individually.

\begin{corollary}\label{corollary:ccNMC}
    Let $G$ be a fully dynamic simple graph, and let $C$ be a connected component of $G$ that has $n_C$ vertices and minimum degree $\delta>0$. Then, the data structure of Theorem~\ref{theorem:main}, can output an NMC sparsifier of $C$ that w.h.p. has $O(n_C/\delta)$ vertices and $O(n_C)$ edges. The update and query time guarantees are the same as in Theorem~\ref{theorem:main}, except that $``n"$ is replaced by $``n_C"$.
\end{corollary}
\begin{proof}
This is an immediate consequence of our whole approach for maintaining an NMC sparsifier. Specifically, it is sufficient to run the construction of the NMC sparsifier on the specified connected component $C$ of the graph. The important observation is that Propositions~\ref{proposition:kforests} and \ref{proposition:contract}, when applied on $C$, take time proportional to its size. 
\end{proof}

\section{Applications}
\label{sec:applications}

\subsection{A Cactus Representation of All Minimum cuts in Dynamic Graphs}
\label{sec:computingcactus}

It is well-known that a graph with $n$ vertices has $O(n^2)$ minimum cuts, all of which can be represented compactly with a data structure of $O(n)$ size that has the form of a cactus graph (see \cite{DBLP:books/cu/NI2008} for details).
As a first immediate application of our main Theorem~\ref{theorem:main}, we show how the NMC sparsifier $\hatG$ of any fully dynamic simple graph $G$ can be used to quickly compute this cactus representation.

\cactustheorem*
\begin{proof}
We initialize the data structure and follow the updates as in Theorem~\ref{theorem:main}.
To answer a query, we first construct the NMC sparsifier $\hatG$ and apply the algorithm from \cite{DBLP:conf/soda/KargerP09} to construct a cactus representation $\mathcal{R}$ for the minimum cuts of $\hatG$. This takes $\tildeO(|\hatG|) = \tildeO(n)$ time and as a byproduct, we also get the value $\lambda$ of the edge connectivity of $\hatG$. Now we follow the same procedure as in \cite{DBLP:journals/jacm/KawarabayashiT19}, to derive a cactus for $G$ from $\mathcal{R}$. We distinguish three cases:

If $\lambda<\delta$, then $\mathcal{R}$ is a cactus for the minimum cuts of $G$. If $\delta<\lambda$, then all minimum cuts of $G$ are trivial. In this case a cactus for $G$ has the form of a star, where the central node represents all vertices of $G$ with degree $>\delta$, the remaining nodes correspond to the vertices with degree $\delta$, and every edge of the cactus is associated with the set of edges incident to the corresponding vertex of $G$. Finally, if $\lambda=\delta$, then we possibly have to enrich $\mathcal{R}$ with more nodes that correspond to the vertices of $G$ with degree $\delta$. Every such vertex $v$ is mapped by the cactus map of $\mathcal{R}$ to some unique node $N$ of $\mathcal{R}$. If $N$ represents more than one vertex of $G$, then we just have to add a new node $\bar{v}$ to $\mathcal{R}$, set the cactus mapping of $v$ to $\bar{v}$, and connect $\bar{v}$ to $N$ with two edges that correspond to the set of edges incident to $v$ in $G$. This completes the method by which we can derive a cactus for $G$ from $\mathcal{R}$.
\end{proof}

To obtain just a single minimum cut of $G$, one can apply the deterministic minimum cut algorithms of \cite{HenzingerRW20,DBLP:journals/jacm/KawarabayashiT19} on $\hatG$, which yields a minimum cut $C$ of $\hatG$ in time $\tildeO(|\hatG|) = \tildeO(n)$. To transform $C$ into a minimum cut of $G$, we compare its size with the minimum degree $\delta$ of any node in $G$:
If $|C|\leq\delta$, then we get a minimum cut of $G$ by simply mapping the edges from $C$ back to the corresponding edges of $G$. Otherwise, a minimum cut of $G$ is given by any vertex of $G$ that has degree $\delta$ (which is easy to maintain throughout the updates).

\subsection{Computing the Maximal k-Edge-Connected Subgraphs}
\label{sec:computingmaximal}

The data structure of Theorem~\ref{theorem:main} can also be used to improve the time bounds for computing the maximal $k$-edge-connected subgraphs of a simple graph, in for cases where $k$ is a sufficiently large polynomial of $n$. Specifically, we get an improvement for $k=\Omega(n^{1/8})$, c.f. Section~\ref{sec:introduction}.

A general strategy for computing these subgraphs is the following.
Let $G$ be a simple graph with $n$ vertices and $m$ edges, and let $k$ be a positive integer. The basic idea is to repeatedly find and remove cuts with value less than $k$ from $G$. First, as long as there are vertices with degree less than $k$, we repeatedly remove them from the graph. Now we are left with a (possibly disconnected) graph where every non-trivial connected component has minimum degree at least $k$. If we perform a minimum cut computation on a non-trivial connected component $S$ of $G$, there are two possibilities: either the minimum cut is at least $k$, or we will have a minimum cut $C$ with value less than $k$. In the first case, $S$ is confirmed to be a maximal $k$-edge-connected subgraph of $G$. In the second case, we remove $C$ from $S$, and thereby split it into two connected components $S_1$ and $S_2$. Then we recurse on both $S_1$ and $S_2$.
Since the graph is simple and $S$ has minimum degree at least $k$, it is a crucial observation that both $S_1$ and $S_2$ contain at least $k$ vertices (see e.g. \cite{DBLP:journals/jacm/KawarabayashiT19}).
Therefore the number of nodes decreases by at least $k$ with every iteration and hence the total recursion depth is $O(n/k)$.

The minimum cut computation takes time $T_{\text{mc}} = \tilde{O}(m)$ \cite{DBLP:journals/jacm/Karger00}, hence the worst-case running time of this approach is $\Theta(n/k \cdot T_{\text{mc}}) = \tilde{O}(mn/k)$.
We can use Theorem~\ref{theorem:main} to bring the time down to $\tildeO(m+n^2/k)$ w.h.p.

\maxkedgesubgraphs*
\begin{proof}
    We apply the algorithm shown in Figure~\ref{algorithm:computemaximal}, but we only perform the minimum cut computations on the NMC sparsifiers of the non-trivial connected component of $G$, as output by the data structure of Corollary~\ref{corollary:ccNMC}. Thus, we view $G$ as a dynamic graph, and first initialize required data structures. Whenever a cut is made, we consider the corresponding edges deleted.
    Since amortized time guarantees are sufficient here, we can use Case~\ref{theorem:main-amortized-adaptive} to get an update time of $\tildeO(1)$.
    The recursion depth is $O(n/k)$, and since the minimum cut computations at every level of the recursion are performed on a collection of graphs with $O(n)$ edges w.h.p., we obtain the theorem.
\end{proof}

\begin{algorithm}[t]
\caption{\textsf{Compute the maximal $k$-edge-connected subgraphs of a simple graph $G$}}
\label{algorithm:computemaximal}
\LinesNumbered
\DontPrintSemicolon
initialize the data structure of Theorem~\ref{theorem:main} on $G$; use the data structures by Holm~et~al.~\cite{DBLP:journals/jacm/HolmLT01} in order to implement DCS and DSF\;
mark every connected component of $G$ as ``active''\;
\ForEach{active connected component $S$ of $G$}{
  \While{$S$ contains a vertex $v$ with degree less than $k$}{
    remove $v$ from $S$, and collect it as a trivial maximal $k$-edge-connected subgraph of $G$\;
  }
  get the contracted graph $\widehat{S}$ from $S$ using Theorem~\ref{theorem:main}\;
  compute a minimum cut $C$ of $\widehat{S}$ using Karger's minimum cut algorithm~\cite{DBLP:journals/jacm/Karger00}\;
  \If{$|C|<k$}{
    remove $C$ from $G$\;
  }
  \Else{
    collect $G[S]$ as a maximal $k$-edge-connected subgraph of $G$, and mark $S$ as an ``inactive'' component\; 
  }
}
\textbf{return} the subgraphs of $G$ that were collected during the course of the algorithm
\end{algorithm}

Through a reduction to simple graphs, Theorem~\ref{theorem:computingmaximal} implies Corollary~\ref{corollary:mult}, which is a similar result with slightly worse time bounds for undirected \emph{multigraphs}.
Finally, the method that establishes Theorem~\ref{theorem:computingmaximal} naturally extends to the case of fully dynamic simple graphs, which yields Theorem~\ref{theorem:dynamicsubgraphs}.

\begin{lemma}
\label{lemma:reduction}
 Let $\mathcal{A}$ be an algorithm that takes $T(n,m,k)$ time to compute the maximal $k$-edge-connected subgraphs of a simple graph with $n$ vertices and $m$ edges.
Then there is an algorithm that takes $T\big(O(kn),m+O(k^2n),k\big)$ time to compute the maximal $k$-edge-connected subgraphs of a \emph{multigraph} graph with $n$ vertices and $m$ edges.
\end{lemma}

\begin{proof}
Let $G$ be a multigraph with $n$ vertices and $m$ edges. We may assume w.l.o.g. that for every two vertices of $G$ there are less than $k$ parallel edges that connect them (because otherwise we can contract every such pair of vertices, since they belong to the same maximal $k$-edge-connected subgraph). Now we replace every vertex $v$ of $G$ with a clique $K(v)$ with $k+1$ vertices, and we use those cliques in order to substitute the parallel edges. To be precise, if $v$ and $u$ are two vertices of $G$ that are connected with $l$ parallel edges, then we select $l$ distinct vertices $x_1,\dots,x_l$ from $K(v)$ and $l$ distinct vertices $y_1,\dots,y_l$ from $K(u)$, and we add $l$ edges of the form $(x_i,y_i)$ for $i\in\{1,\dots,l\}$. Thus, we get a simple graph $G'$ with $O(kn)$ vertices and $m+O(k^2n)$ edges, that essentially has the same maximal $k$-edge-connected subgraphs as $G$.
The result follows by applying $\mathcal{A}$ on $G'$.
\end{proof}

\section*{Acknowledgements}
\textit{Monika Henzinger and Evangelos Kosinas}: This project has received funding from the European Research Council (ERC) under the European Union's Horizon 2020 research and innovation programme (MoDynStruct, No. 101019564)  \includegraphics[width=0.9cm]{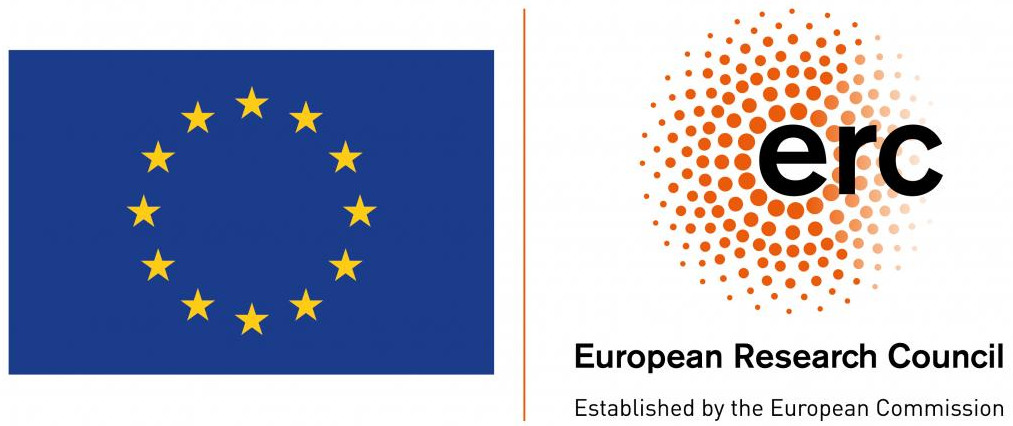} and the Austrian Science Fund (FWF) grant  \href{https://www.doi.org/10.55776/Z422}{DOI~10.55776/Z422} and grant \href{https://www.doi.org/10.55776/I5982}{DOI~10.55776/I5982}.\\
\textit{Harald Räcke and Robin Münk}: This project has received funding from the Deutsche Forschungsgemeinschaft (DFG, German Research Foundation) – 498605858.

\bibliography{main}

\appendix

\section{Sampling Elements from Dynamic Lists}
\label{sec:sampling}
Let $G$ be a dynamic graph, and let $n$ denote the number of vertices of $G$ at the current moment.
In order to construct a random $2$-out contraction of $G$, we need to sample two times an incident edge to $v$ (uniformly at random, with repetition allowed), for every vertex $v$ of $G$. If the adjacency list of $v$ was fixed, then we could perform this sampling of incident edges to $v$ in constant time per sampling (assuming that there is a random number generator that can provide a random number between $1$ and $n$ in constant time per query). This works by representing the adjacency list of $v$ as an array, and then returning the edge in the $t$-th position of this array, where $t$ is a randomly generated number between $1$ and $\mathit{deg}(v)$. However, since $G$ is a dynamic graph, the adjacency list of $v$ may change, and this way to sample will not work, because we cannot represent the adjacency list of $v$ as a static array. 

Here we provide a data structure that enables us to perform the sampling of an incident edge to any vertex $v$ in worst-case $O(\log{n})$ time per query, while supporting the update of the adjacency list of $v$ in worst-case constant time per update. The idea is reminiscent of the min-heap data structure (see e.g. \cite{DBLP:books/daglib/0023376}), although for us there is no min-heap property to maintain, and also we want to avoid using arrays, in order to provide clear worst-case time guarantees. Thus, we represent the adjacency list of $v$ both as a doubly linked list and as a binary tree. Specifically, let $e_1,\dots,e_l$ be the edges currently in the adjacency list $L$ of $v$, in this order. Then, we have a binary tree $T$ with root $e_1$, and the children of any edge $e_i\in\{e_1,\dots,e_l\}$ are $e_{2i}$ and $e_{2i+1}$ (whenever $2i\leq l$ or $2i+1\leq l$). We call $e_{2i}$ the \emph{left} child of $e_i$, and $e_{2i+1}$ the \emph{right} child of $e_i$. Notice that $T$ has $\lceil \log_2(l+1) \rceil$ levels, and at every level the edges appear in increasing order from left to right w.r.t. their order in $L$. Thus, the sampling of an edge in $L$ can be performed as follows. We generate a random number $t$ between $1$ and $l$ (we assume that we maintain in a variable $l$ the size of the adjacency list of $v$). If $t=1$, then we return $e_1$ (the root of $T$). Otherwise, let $1a_1\dots a_b$ be the binary number representation of $t$. Then we traverse $T$ starting from the root, according to the digits of $t$ starting from $a_1$: every time we descend to the left or to the right child of the current node depending on whether the current digit is $0$ or $1$ respectively. After exhausting the digits of $t$, we return the edge that corresponds to the final node that we reached with this traversal. Thus, we have returned a random element from $L$ in $O(\lceil \log_2(l+1) \rceil)=O(\log{n})$ time. 

In order to accommodate for fast (worst-case constant time) updates of the adjacency list of $v$, we have to enrich $T$ with more pointers. First, we organize $T$ in $\lceil \log_2(l+1) \rceil$ levels. We consider $\{e_1\}$ to constitute level $1$, $\{e_2,e_3\}$ to constitute level $2$, and so on. In general, level $d$ contains the edges $e_{2^{d-1}},e_{2^{d-1}+1},\dots,e_N$, where $N=2^d-1$ or less, depending on whether $d$ is the largest level of $T$ or not. For every level $d$ of $T$, there is a doubly linked list $T_d$ that consists of the edges at that level, in increasing order w.r.t. their order in $L$. We also maintain in a variable $\lambda$ the number of levels of $T$. 

Now, in order to perform an insertion of a new edge $e$ to the adjacency list of $v$, we simply append $e$ to the deepest level of $T$, in the last position. To be precise, we first take the last edge $e'\in L$. If $e'$ is the left child of its parent $p$ in $T$, then we let $e$ be the right child of $p$ (and we update the list $T_\lambda$ appropriately). Otherwise, we go to the next element $p'$ after $p$ in $T_{\lambda-1}$. If $p'$ exists, then we let $e$ be the left child of $p'$. Otherwise, $e'$ is the last entry of $T_\lambda$, and so we must create a new level for $T$. Thus, we take the first entry $f$ of $T_\lambda$, we let $e$ be the left child of $f$ on $T$, we set $\lambda\leftarrow\lambda+1$, and we initialize the new list $T_\lambda$ with a single entry for $e$. Thus, the insertion of $e$ demands a constant number of pointer manipulations. The deletion of an edge $e$ from $L$ is simply performed by replacing $e$ on $T$ with the last element of $L$. It is easy to see that this can be done after a constant number of pointer manipulations.  

\section{Reducing DSF and DCS to Dynamic Minimum Spanning Forest}
\label{sec:appendixproof}
In order to prove Lemma~\ref{lemma:reducetoMSF}, we define the version of the dynamic MSF problem that we need. This works on a dynamic weighted graph $G$, and maintains a minimum spanning forest $F$ of $G$. Specifically, the data structure for the dynamic MSF supports the following operations.

\begin{itemize}
\item{$\mathtt{insert}(e,w)$. Inserts a new edge $e$ to $G$ with weight $w$. The forest $F$ is updated accordingly: If $e$ has endpoints in two different trees of $F$, then it is automatically included in $F$. Otherwise, if the tree path of $F$ that connects the endpoints of $e$ contains an edge with weight larger than $w$, then one such edge $e'$ with maximum weight is deleted from $F$, and $e$ is automatically used as a replacement to $e'$. These events (and the edge $e'$) are reported to the user.}
\item{$\mathtt{delete}(e)$. Deletes $e$ from $G$. If $e$ is a tree edge of $F$, it is also deleted from $F$. This separates a tree $T$ of $F$ into two trees, $T_1$ and $T_2$. If $T$ is still connected in $G$, then an edge of $G$ with minimum weight is used as a replacement to $e$ in order to reconnect $T_1$ and $T_2$, and it is reported to the user.}
\end{itemize}

\begin{proof}[Proof of Lemma~\ref{lemma:reducetoMSF}]
Let $G$ be the dynamic input graph. First we will show the easier reduction from DSF to DCS. We need to maintain a spanning forest of $G$, and support the operations $\mathtt{insert}$ and $\mathtt{delete}$. Let us assume that so far we have maintained a spanning forest $F$ of $G$ using the DCS data structure. Suppose that we receive a DSF call $\mathtt{insert}(e)$ for an edge $e$. Then we just call $\mathtt{insert}_G(e)$ and then $\mathtt{insert}_F(e)$ (of the DCS data structure). The second call will insert $e$ on $F$ if the endpoints of $e$ lie in different trees of $F$ (otherwise $F$ remains the same). Thus, $F$ is still a spanning forest of $G\cup\{e\}$. Now suppose that we receive a DSF call $\mathtt{delete}(e)$ for an edge $e$. Let $x$ and $y$ be the endpoints of $e$. Then we first call $\mathtt{delete}_F(e)$. If this operation reports that $e$ is not part of $F$, then we just delete $e$ from $G$ with $\mathtt{delete}_G(e)$. Otherwise, we seek for a replacement of $e$ by calling $\mathtt{find\_cutedge}(x)$. If this operation returns an edge $e'$, then we call $\mathtt{insert}_F(e')$. In any case, we then delete $e$ from $G$ with $\mathtt{delete}_G(e)$. It is easy to see that a spanning forest of $G$ is thus maintained.

Now we will show that the DCS problem can be reduced to dynamic MSF (with an additional $O(\log n)$ worst-case time per operation). For that, we will also need some additional data structures, which we will describe shortly. Let us assume that so far we have maintained the forest $F$ that the user controls with the DCS data structure. We maintain the following invariants for the minimum spanning forest $F'$ maintained by the MSF data structure: $(1)$ the edges of $G$ are partitioned into \emph{light} edges with weight $0$ and \emph{heavy} edges with weight $1$, $(2)$ the light edges are precisely those that constitute $F$, and $(3)$ every tree of $F$ is a subtree of a tree of $F'$. (Notice that $(3)$ actually follows from $(1)$ and $(2)$, and from the fact that $F'$ is a minimum spanning forest of $G$. However, we state it explicitly for convenience.)

Now we will show how to reduce every DCS operation to MSF operations (plus some additional ones). For that, we assume that we maintain a copy of $F'$ with the dynamic tree data structure of Sleator and Tarjan~\cite{DBLP:journals/jcss/SleatorT83} (ST). This data structure maintains a collection of rooted trees, and it allows to reroot them, link them, cut them, store values on tree edges, and update those values, all in worst-case $O(\log n)$ time per operation (where $n$ is the total number of vertices). We note that the purpose of rerooting is to make a specified vertex $v$ the root of the tree that contains it. (Notice that this may change some parent pointers accordingly, but the cleverness of the data structure is that it does all that implicitly, and can still perform every update, and report the correct answer to every query, in $O(\log{n})$ time.) Furthermore, given a vertex $v$ that is not the root $r$ of the tree that contains it, it can report in worst-case $O(\log n)$ time the edge of the tree path from $r$ to $v$ that has maximum weight and is closest to the root. Let us call this operation $\mathtt{max\_edge}(v)$. Thus, if we have two vertices $x$ and $y$ on the same tree, and we want to find the edge of maximum weight on the tree path from $x$ to $y$ that is closest to $x$, then we first reroot at $x$, and then call $\mathtt{max\_edge}(y)$.

We let the forest maintained by the ST data structure be an exact copy of $F'$. (Thus, whenever the MSF data structure updates $F'$, the same changes are mimicked by the ST data structure.) We also maintain a copy of the collection of trees in $F'$ using the Euler-tour (ET) data structure \cite{DBLP:conf/stoc/HenzingerK95}. The ET data structure supports updates to the trees such as linking and cutting in worst-case $O(\log n)$ time per operation, and it also provides aggregate information for every tree $T$, such as the maximum weight of an edge in $T$ (and also a pointer to such an edge), in worst-case $O(\log n)$ time per query. 

Now, the DCS operation $\mathtt{insert}_G(e)$ is simulated by the MSF operation $\mathtt{insert}(e,1)$. Notice that this does not violate our invariants. Now suppose that we receive a call $\mathtt{insert}_F(e)$. First, we have to check whether the endpoints $x$ and $y$ of $e$ are in the same tree of $F$. It is easy to do that, assuming that we maintain e.g. a copy of $F$ using another ST data structure. If $x$ and $y$ belong to the same tree of $F$, then we do nothing. Otherwise, we insert $e$ to $F$, but now we also have to maintain our invariants on $F'$. To do this, we just have to re-insert $e$ to $G$ with the MSF data structure, but this time we will insert it with weight $0$, so that it will be forced to become part of $F$. Thus, we call $\mathtt{delete}(e)$ and then $\mathtt{insert}(e,0)$. It is easy to see that all invariants are maintained thus (in particular, $(3)$ remains true because it held so for the trees of $F$ that contained the endpoints of $e$ before its insertion). 

For $\mathtt{delete}_G(e)$, we just delete $e$ with $\mathtt{delete}(e)$. If $e$ was an edge of $F$, then it is automatically deleted from $F'$ (in this case, $e$ was indeed also an edge of $F'$, due to invariant $(2)$), and the invariants are still maintained. For $\mathtt{delete}_F(e)$, we simply have to convert $e$ into a heavy edge of $G$ (so that it is no longer interpreted as part of $F$). To do this, we just delete it from $G$ with  $\mathtt{delete}(e)$, and re-insert it with $\mathtt{insert}(e,1)$.

Finally, we will show how to answer every query of the form $\mathtt{find\_cutedge}(v)$, for a vertex $v$ of $G$. Recall that this operation must return an edge of $G$ that has one endpoint on the tree $T$ of $F$ that contains $v$, and the other endpoint outside of $T$ (if such an edge exists). Due to invariant $(3)$, we have that $T$ is a subtree of a tree $T'$ of $F'$. Therefore, since $F'$ is a minimum spanning forest of $G$, there is an edge that connects $T$ and $V(G)\setminus T$ if and only if $T'\neq T$. It is easy to test whether $T'=T$ by checking whether the number of vertices of the tree of $F$ that contains $v$ is the same as that of the tree of $F'$ that contains $v$, using e.g. the ST data structures. So let us assume that $T'\neq T$. Then, due to our invariants, there is a tree edge on $T'$ with weight $1$ that connects $T$ and $V(G)\setminus T$.
In order to find such an edge, we first ask the ET data structure on $T'$ to give us an edge $e'$ that has maximum weight. (Thus, due to $(1)$ and $(2)$, we have that $e'$ has weight $1$ and is not part of $T$.) If one of the endpoints of $e'$ is on $T$, then we can simply return $e'$. Otherwise, let $x$ and $y$ be the endpoints of $e'$. Now we use the ST data structure in order to reroot $T'$ on $v$, and then we use ST again in order to determine which of $x$ and $y$ is the parent of the other on the (rerooted) $T'$. So let us assume w.l.o.g. that $x$ is the parent of $y$ (and so $x$ is closer than $y$ to the root $v$). Then, using ST once more, we ask for the edge with maximum weight on the tree path of $T'$ from $v$ to $x$ that is closest to $v$ using $\mathtt{max\_edge}(x)$. Notice that, due to our invariants, this is precisely an edge that has an endpoint on $T$, and the other endpoint on $V(G)\setminus T$. This concludes the description of our reduction from DCS to MSF.
\end{proof} 

\end{document}